\documentclass[letterpaper, 10 pt, conference]{ieeeconf}  

\usepackage{graphicx}
\usepackage{amsmath}
\usepackage{mathtools}
\usepackage{amssymb}
\usepackage{amsfonts}
\usepackage{upgreek}
\usepackage{hyperref}
\usepackage{cite}
\usepackage{algorithm}
\usepackage{algorithmic}
\usepackage{mathdots}
\usepackage{bm}
\usepackage{textcomp}
\usepackage{cases}
\usepackage[font=footnotesize]{caption}
\usepackage[font=footnotesize]{subcaption}
\usepackage{color}
\usepackage{soul}
\usepackage[normalem]{ulem}
\usepackage{graphbox}
\usepackage{hyperref}
\usepackage{wasysym}
\usepackage{xcolor}
\usepackage{tikz}
\usepackage{pgfplots}
\pgfplotsset{compat=1.18}
\usetikzlibrary{intersections}
\usetikzlibrary{patterns}
\usepgfplotslibrary{fillbetween}
\usetikzlibrary{arrows.meta}
\usepackage{comment}

\newtheorem{lemma}{Lemma}
\newtheorem{remark}{Remark}
\newtheorem{theorem}{Theorem}
\newtheorem{assumption}{Assumption}

\newtheorem{example}{Example}

\newcommand{\R}{\mathbb{R}}
\newcommand{\Rpos}{\mathbb{R}_{\geq 0}}
\newcommand{\ru}{\bar{\rho}}
\newcommand{\rl}{\underline{\rho}}
\newcommand{\dru}{\dot{\bar{\rho}}}
\newcommand{\drl}{\dot{\underline{\rho}}}
\newcommand{\rank}{\mathrm{rank}}
\newcommand{\alphb}{\bar{\alpha} }
\newcommand{\alphbopt}{\bar{\alpha}_{\text{opt}} }
\newcommand{\alphaopt}{\alpha_{\text{opt}} }
\newcommand{\epialph}{\varepsilon_{\alpha}}
\newcommand{\epialphb}{\bar{\varepsilon}_{\alpha}}
\newcommand{\depialph}{\dot{\varepsilon}_{\alpha} }
\newcommand{\xialph}{\xi_{\alpha} }
\newcommand{\alphah}{\hat{\alpha} }
\newcommand{\rlalph}{\underline{\rho}_{\alpha}}
\newcommand{\rualph}{\bar{\rho}_{\alpha}}

\newcommand{\C}{\mathcal{C}}

\newcommand{\T}{\mathcal{T}}
\newcommand{\I}{\mathcal{I}}

\newcommand{\Hes}{\mathcal{H}}

\newcommand{\Ob}{\bar{\Omega}}
\newcommand{\Ox}{\Omega_x}
\newcommand{\Oxs}{\Omega_x^s}
\newcommand{\Oz}{\Omega_z}

\newcommand{\mathdot}{\mathord{\cdot}}
\newcommand{\Oalphh}{\Omega_{\alphah}}
\newcommand{\rhos}{\rho_s}
\newcommand{\rhod}{\rho_d}
\newcommand{\drhos}{\dot{\rho}_s}
\newcommand{\drhod}{\dot{\rho}_d}
\newcommand{\cl}{\mathrm{cl}}
\newcommand{\gradxalph}{\nabla_x \alpha}
\newcommand{\rondalph}{\frac{\partial \alpha}{\partial x}}

\newcommand{\col}{\mathrm{col}}

\newcommand{\taum}{\tau_{\mathrm{max}}}

\def\@IEEEtablestring{table}

\long\def\@makecaption#1#2{%
	\ifx\@captype\@IEEEtablestring%
	\begin{center}{\footnotesize #1}\\{\footnotesize\scshape #2}\end{center}%
	\@IEEEtablecaptionsepspace
	\else
	\@IEEEfigurecaptionsepspace
	\setbox\@tempboxa\hbox{\footnotesize #1.~~ #2}%
	\ifdim \wd\@tempboxa >\hsize%
	\setbox\@tempboxa\hbox{\footnotesize #1.~~ }%
	\parbox[t]{\hsize}{\footnotesize \noindent\unhbox\@tempboxa#2}%
	\else%
	\ifcenterfigcaptions \hbox to\hsize{\footnotesize\hfil\box\@tempboxa\hfil}%
	\else \hbox to\hsize{\footnotesize\box\@tempboxa\hfil}%
	\fi\fi\fi}

\IEEEoverridecommandlockouts                              
\title{\LARGE \bf
Control of Nonlinear Systems Under Multiple Time-Varying Output Constraints: A Single Funnel Approach \\ (extended version)
}

\author{Farhad Mehdifar, Lars Lindemann, Charalampos P. Bechlioulis, and Dimos V. Dimarogonas
\thanks{This work is supported by ERC CoG LEAFHOUND, the KAW foundation, and the Swedish Research Council (VR).}
\thanks{F. Mehdifar and D. V. Dimarogonas are with the Division of Decision and Control Systems, KTH Royal Institute of Technology, Stockholm, Sweden.   {\tt\small mehdifar@kth.se; dimos@kth.se}}%
\thanks{Lars Lindemann is with is with Thomas Lord Department of Computer Science, University of Southern California, Los Angeles, CA, USA.
	{\tt\small llindema@usc.edu}}
\thanks{C. P. Bechlioulis is with the Division of Systems and Control of the Department of Electrical and Computer Engineering at University of Patras, Patra, Greece. {\tt\small chmpechl@upatras.gr}}%
}

\begin{document}

\maketitle
\thispagestyle{empty}
\pagestyle{empty}

\begin{abstract}
	This paper proposes a novel control framework for handling (potentially coupled) multiple time-varying output constraints for uncertain nonlinear systems. First, it is shown that the satisfaction of multiple output constraints boils down to ensuring the positiveness of a scalar variable (the signed distance from the time-varying output-constrained set's boundary). Next, a single funnel constraint is designed properly, whose satisfaction ensures convergence to and invariance of the time-varying output-constrained set. Then a robust and low-complexity funnel-based feedback controller is designed employing the prescribed performance control method. Finally, a simulation example clarifies and verifies the proposed approach.
\end{abstract}
\section{Introduction}

Due to practical needs and theoretical challenges, controlling nonlinear systems under constraints has attracted much attention in the past decade. In particular, imposing time-varying output constraints in nonlinear control systems is motivated by ensuring tracking/stabilization performance or safety requirements \cite{tee2011control, ilchmann2002tracking, berger2021funnel, bechlioulis2014low, bechlioulis2008robust, theodorakopoulos2015low, jin2018adaptive,  xu2018constrained}. Existing feedback control approaches that can deal with time-varying output constraints are mainly categorized as control designs based on Time-Varying Barrier Lyapunov Function (TVBLF) \cite{tee2011control}, Funnel Control (FC) \cite{ilchmann2002tracking,berger2021funnel}, and Prescribed Performance Control (PPC) \cite{bechlioulis2008robust,bechlioulis2014low} methods, where FC and PPC offer simple and more constructive control designs with inherent robustness. 

The aforementioned control methods are often used for ensuring a user-defined transient and steady-state performance on the tracking/stabilization error by confining its evolution within a user-defined time-varying funnel as the only output constraint. For nonlinear systems with multiple time-varying output constraints (funnels), TVBLF, FC, and PPC methods are only applied when the system outputs are selected as independent states of the system with the number of inputs equal to the number of outputs \cite{bechlioulis2008robust, theodorakopoulos2015low, jin2018adaptive}. In particular, such a choice for the system's outputs ensures that the resulting time-varying output constraints (funnels) remain decoupled at all times. Since various applications (e.g., those dealing with safety \cite{glotfelter2017nonsmooth} and general spatiotemporal specifications \cite{lindemann2021funnel}) require considering arbitrary (potentially coupled) multiple time-varying output constraints, it is significant to develop control methods for uncertain nonlinear systems under more general classes of output constraints.

In this paper, inspired by the works \cite{lindemann2017prescribed,lindemann2021funnel}, we propose an alternative control design method for uncertain nonlinear systems with potentially coupled multiple time-varying output constraints, where we encapsulate the satisfaction of all time-varying output constraints by imposing the positiveness of a single scalar variable (the signed distance from the time-varying output constrained set's boundary). We show that by an appropriate design of a single funnel constraint for the mentioned scalar variable, one can enforce its positiveness in a finite time, which leads to the satisfaction of all time-varying output constraints. To this end, we employ the PPC method to design a robust low-complexity control law for uncertain nonlinear systems. Here by a low-complexity control design, we mean that no estimation scheme is used in the proposed control law.

We highlight that one can also employ time-varying Control Barrier Functions (CBFs) \cite{xu2018constrained,ames2019control} for controlling nonlinear systems under time-varying output constraints. Indeed, our work captures multiple time-varying output constraints similarly to \cite{glotfelter2017nonsmooth}, in which the composition of multiple time-invariant CBFs is considered. Nevertheless, typical control synthesis using the CBF notion requires exact knowledge of the system dynamics and solving online QP optimization problems. In contrast, in this work, we provide a computationally tractable (i.e., optimization-free) and robust (model-free) control law.

In contrast to TVBLF, FC, and PPC methods, in which system outputs are merely subjected to (often symmetric) funnel constraints, we consider generic asymmetric funnel constraints as well as one-sided constraints over the system outputs (see Section \ref{sec:problemformulation}). Furthermore, while the aforementioned control methods require the satisfaction of all output constraints at the initial time, our proposed control method addresses convergence to the time-varying output-constrained set within a user-defined finite time in case the output constraints are not initially satisfied. In this respect, we ensure convergence to and invariance of the time-varying output-constrained set within an appointed finite time. Overall, in this work, we generalize feedback control designs for nonlinear systems under an expanded class of time-varying output constraints and facilitate the controller synthesis and stability analysis. In this respect, reference tracking under prescribed transient and steady-state specifications and time-invariant output constraints for nonlinear systems fit into our results as special cases.

\section{Problem Formulation} 
\label{sec:problemformulation}

Consider the following first-order nonlinear input affine dynamical system:
\begin{equation} \label{eq:sys_dynamics_firstorder}
	\begin{cases}
		\dot{x} = f(x) + g(x) u + w(t), \\
		y = h(x),
	\end{cases}
\end{equation}
where $x = \col(x_i) \coloneqq [x_1, x_2, \ldots, x_n]^\top \in \R^{n}$ is the system's state, $u \in \R^n$ and  $y \coloneqq \col(y_i) = [y_1, y_2, \ldots, y_m]^\top \in \R^{m}$ denote the control input and the system's output vector, respectively. Moreover, $f:  \R^n \rightarrow \R^{n}$ and $g: \R^n \rightarrow \R^{n \times n}$ are locally Lipschitz continuous in $x$, and $h: \R^n \rightarrow \R^{m}$ is $\C^2$ (i.e., two times continuously differentiable). In particular, let $h(x) = [h_1(x), h_2(x), \ldots, h_m(x)]^{\top}$, so that $y_i =h_i(x), i \in \I \coloneqq \{1, \ldots, m\}$. Furthermore, $w: \R_{\geq 0} \rightarrow \R^n$ denotes bounded piecewise continuous external disturbances, where $\|w(t)\| \leq \bar{w}, \forall t \geq 0$ ($\|\mathdot\|$ denotes the Euclidean norm). In addition, let $x(t;x_0)$ indicate the solution of \eqref{eq:sys_dynamics_firstorder} under the initial condition $x_0 \coloneqq x(0)$ and control input $u$.

\begin{assumption} \label{assum:uncertain}
	The functions $f(x)$ and $g(x)$, and the constant $\bar{w}$ are unknown for the controller design.
\end{assumption}

\begin{assumption} \label{assum:symm_g}
	The symmetric component of the input gain matrix $g(x)$, i.e., $g_s(x) \coloneqq \frac{1}{2}(g^{\top}(x)+g(x))$, is uniformly sign definite in $x$. Without loss of generality, we assume $g_s(x)$ is uniformly positive definite in $x \in \R^n$, i.e., $z^{\top} g_s(x) z > 0, \forall z,x \in \R^n, z \neq 0$.
\end{assumption}

Note that Assumption \ref{assum:symm_g} constitutes a controllability condition on \eqref{eq:sys_dynamics_firstorder}.

Let the outputs of \eqref{eq:sys_dynamics_firstorder} be subjected to the following class of time-varying constraints:
\begin{equation}\label{eq:outputconst}
\!\!\! \rl_i(t) < h_i(x) < \ru_i(t), \;\; i \in  \I = \{1, \ldots, m\}, \;\; \forall t \geq 0, \!
\end{equation}
where $\rl_i, \ru_i: \Rpos \rightarrow \R \cup \{\pm \infty\}, i \in \I$. We assume for each $i \in \I$, that \textit{at least} one of $\ru_i(t)$ and $\rl_i(t)$ is a bounded $\C^1$ function of time with a bounded derivative. In other words, we allow $\rl_i(t) = - \infty$ (resp. $\ru_i(t) = + \infty$) when $\ru_i(t)$ (resp. $\rl_i(t)$) is bounded for all $t \geq 0$. In this respect, \eqref{eq:outputconst} can either represent \textit{Lower Bounded One-sided} (LBO) time-varying constraints in the form of $\rl_i(t) < h_i(x)$, \textit{Upper Bounded One-sided} (UBO) time-varying constraints in the form of $h_i(x) < \ru_i(t)$, as well as (time-varying) \textit{funnel constraints} in the form of $\rl_i(t) < h_i(x) < \ru_i(t)$, for which both $\ru_i(t)$ and $\rl_i(t)$ are bounded. Without loss of generality, we assume that the first $p$ constraints in \eqref{eq:outputconst}, i.e., for $i = \{1,\ldots, p\}, 0 \leq p \leq m$, are funnel constraints, $q$ LBO constraints are indexed by $i = \{p+1,\ldots, q\}, 0 \leq q \leq m - p$ in \eqref{eq:outputconst}, and the remaining $m-p-q$ constraints represent UBO constraints for which $i = \{p+q+1,\ldots, m\}$ in \eqref{eq:outputconst}. We also assume that each funnel constraint is well-defined in the sense that for each $i = \{1, \ldots, p\}$, there exists a $\epsilon_i > 0$ such that $\ru_i(t) - \rl_i(t) \geq \epsilon_i, \forall t\geq0$, which indicates that the $p$ funnel constraints in \eqref{eq:outputconst} are \textit{separately feasible}.

\begin{remark}
	In our problem formulation, requirements like regulation and tracking translate to fulfilling specific types of output constraints. For instance, take \eqref{eq:sys_dynamics_firstorder} with $h(x) = x$, where $x \in \R$. Let $x_{d}(t)$ be a continuously differentiable and bounded reference signal, whose derivative is also bounded. The tracking goal can be achieved by ensuring $-\rho_d(t) < x - x_d(t) < \rho_d(t)$, with $\rho_d(t)$ being a continuously differentiable and bounded signal that approaches a small neighborhood of zero (akin to tracking under prescribed performance \cite{bechlioulis2008robust}). This requirement can equivalently be expressed as $-\rho_d(t) + x_d(t) < x < \rho_d(t) + x_d(t)$, serving as a time-varying output constraint for \eqref{eq:sys_dynamics_firstorder}.
\end{remark}

Define next $\Ob(t)$ based on \eqref{eq:outputconst} as:
\begin{equation} \label{eq:omeg_x_t}
	\Ob(t) \coloneqq \{ x \in \R^n \mid  \rl_i(t) < h_i(x) < \ru_i(t), i \in \I\}.
\end{equation}

In this paper, our goal is to design a low-complexity continuous robust feedback control law $u(t,x)$ for \eqref{eq:sys_dynamics_firstorder} such that the closed-loop system trajectories satisfy the time-varying output constraints \eqref{eq:outputconst} $\forall t > T \geq 0$, where $T$ is a user-defined finite time after which the output constraints are surely satisfied (i.e., $x(t;x_0) \in \Ob(t), \forall t > T \geq 0$). Note that this problem  reduces to establishing only invariance of $\Ob(t)$ for all $t \geq 0$, if $x(0) \in \Ob(0)$ ($T = 0$). On the other hand, having $x(0) \notin \Ob(0)$ indicates establishing: (i) finite-time convergence to $\Ob(T)$, and (ii) ensuring invariance of $\Ob(t)$, for all $t > T$. 

\section{Main Results}

\subsection{A Scalar Variable for Constraints Satisfaction}
\label{subsec:scalar_metric}
Here, inspired by works in \cite{lindemann2017prescribed,lindemann2021funnel}, where a funnel-based control design is developed for handling Signal Temporal Logic (STL) specifications for nonlinear systems, we present the signed distance from the boundary of the time-varying output constrained set as a useful scalar variable that encodes checking both feasibility and satisfaction of the constraints.

Notice that the $m$ output constraints in \eqref{eq:outputconst} can be re-written in the following format:
	\begin{subequations}\label{eq:predicate_const_rep}
	\begin{align} 
		&\resizebox{.93\hsize}{!}{$\!\!\!\!\! \begin{cases}
				\psi_{2i-1}(t,x) = h_i(x) - \rl_i(t) > 0, \; \text{(funnel constraints)} \\
				\psi_{2i}(t,x) = \ru_i(t) - h_i(x) > 0,  \;  i \in \{1,\ldots,p\}
			\end{cases}$} \!\!\!\!\!\!\!\!\!\!\!\! \label{eq:predicate_funnel_rep} \\
		&\resizebox{.93\hsize}{!}{$\!\!\!\!\! \begin{cases}
				\psi_{i}(t,x) = h_j(x) - \rl_j(t) > 0, \quad \text{(LBO constraints)} \\
				\; i \in \{2p+1,\ldots, 2p+q\}, \; j \in \{p+1,\ldots, p+q\}, \\
				\psi_{i}(t,x) = \ru_j(t) - h_j(x) > 0, \quad \text{(UBO constraints)} \\ 
				\; i \in \{2p+q+1,\ldots, m+p\}, \;  j \in \{p+q+1,\ldots, m\}.
			\end{cases}$} \!\!\!\!\!\!\!\!\!\!\!\! \label{eq:predicate_onesided_rep}   
	\end{align}
\end{subequations}
Now, without loss of generality, consider all these $m+p$ constraints in \eqref{eq:predicate_const_rep} as:
\begin{equation} \label{eq:psi_constr}
	\psi_i(t,x) > 0, \quad i \in \I_{\psi} \coloneqq \{1,\ldots,m+p\},
\end{equation}
where $\psi_i: \R_{\geq 0} \times \R^{n} \rightarrow \R$ are $\C^2$ in $x$ and $\C^1$ in $t$. As a result, one can re-write \eqref{eq:omeg_x_t} as: 
\begin{equation} \label{eq:omega_y}
	\Ob(t) = \{ x \in \R^n \mid \psi_i(t,x) > 0, \forall i \in \I_{\psi} \}.
\end{equation}

Now define the scalar function $\alphb: \R_{\geq 0} \times \R^n \rightarrow \R$, as:
\begin{equation} \label{eq:metric}
	\alphb(t,x) \coloneqq \min \{ \psi_1(t,x), \ldots , \psi_{m+p}(t,x) \},
\end{equation}
where $\alphb(t,x)$ represents the signed (minimum) distance from the boundary of $\cl(\Ob(t))$, which is the closure of the time-varying output constrained set $\Ob(t)$ in \eqref{eq:omega_y}. In this respect, one can re-write \eqref{eq:omega_y} as the zero super level set of $\alphb(t,x)$:
\begin{equation} \label{eq:omega_alpha_bar}
	\Ob(t) = \{x \in \R^n \mid \alphb(t,x) > 0 \}.
\end{equation}
Note that if $\alphb(t^{\prime},x)<0$, then \textit{at least} one constraint is not satisfied at $t = t^{\prime}$, while $\alphb(t,x)>0, \forall t \geq 0$ means that all constraints are satisfied for all times. Owing to the usage of the $\min$ operator in \eqref{eq:metric}, in general, $\alphb(t,x)$ is a continuous but nonsmooth function, therefore, to facilitate the controller design and stability analysis we will consider the smooth under-approximation of $\alphb(t,x)$ given by the following log-sum-exp function \cite{gilpin2020smooth}:
\begin{equation}\label{smooth_alph}
		\alpha(t,x) \coloneqq -\frac{1}{\nu} \ln \Big( \sum_{i=1}^{m+p} e^{- \nu  \, \psi_i(t,x)} \Big) 
		\leq \alphb(t,x),
\end{equation}
for which $\alphb(t,x) \leq \alpha(t,x) + \frac{1}{\nu} \ln(m+p)$ holds and $\nu>0$ is a tuning coefficient whose larger values gives a closer (under) approximation (i.e, $\alpha(t,x) \rightarrow \alphb(t,x)$ as $\nu \rightarrow \infty$). Note that, $\alpha(t,x)$ represents the signed distance from the boundary of a \textit{smooth inner-approximation} of $\cl(\Ob(t))$. Therefore, ensuring $\alpha(t,x) > 0, \forall t \geq 0$ guarantees $\alphb(t,x) > 0, \forall t \geq 0$ and thus the satisfaction of \eqref{eq:psi_constr} (equivalently \eqref{eq:outputconst}). Define $\Omega(t) \subset \Ob(t)$ as the smooth inner-approximation of the set $\Ob(t)$, given by:
\begin{equation} \label{eq:omega_alpha}
	\Omega(t) \coloneqq \{x \in \R^n \mid \alpha(t,x) > 0 \},
\end{equation} 
and we have $x \in \Omega(t) \Rightarrow x \in \Ob(t)$. Moreover, when $\Ob(t)$ is bounded, then $\Omega(t)$ will be bounded. Let $\partial \cl(\Ob(t))$ and $\partial \cl(\Omega(t))$ indicate the boundaries of the sets $\cl(\Ob(t))$ and $\cl(\Omega(t))$, respectively. Fig.\ref{fig:example1_cases} depicts snapshots of $\Ob(t)$ and $\Omega(t)$ with $\nu = 2$ in \eqref{smooth_alph} for the following examples:

\begin{example} \label{ex:concave_example}
		 Consider $h(x) = [h_1(x), h_2(x), h_3(x)]^{\top}$, where $h_1(x) = x_1$, $h_2(x) = x_2 - x_1$, and $h_3(x) = 0.3x_1^2 + x_2$ and let the output constraints be $\rl_1(t) < h_1(x) < \ru_1(t)$ (funnel constraint), $\rl_2(t) < h_2(x)$ (LBO constraint), and $h_3(x) < \ru_3(t)$ (UBO constraint), respectively. Fig.\ref{fig:omega_y_bar_and_Omega_y} depicts a snapshot of the time-varying output constrained set and its smooth inner approximation, for which $-\rl_1(t)=\ru_1(t) =2$, $\rl_2(t) = -2$, and $ \ru_3(t)= 4$.  
\end{example}

\begin{example}\label{ex:independent_funnels_example}
	Consider $h(x) = [h_1(x), h_2(x)]^{\top}$, with $h_1(x) = x_1$ and $h_2(x) = 0.3x_1^2 - x_2$ and let the output constraints be $\rl_1(t) < h_1(x) < \ru_1(t)$ and $\rl_2(t) < h_2(x) < \ru_2(t)$ (two funnel constraints), respectively. Fig.\ref{fig:condition II} depicts a snapshot of the time-varying output-constrained set and its smooth inner approximation, for which $\rl_1(t)=-3, \ru_1(t) =2$, and $\rl_2(t) = -3, \ru_2(t) = 1$. 
\end{example}   

\begin{figure}[!tbp]
	\centering
	\begin{subfigure}[t]{0.43\linewidth}
		\centering
		\includegraphics[width=\linewidth]{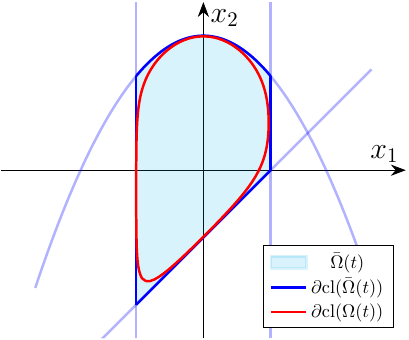}
		\caption{}
		\label{fig:omega_y_bar_and_Omega_y}
	\end{subfigure}%
	~
	\begin{subfigure}[t]{0.43\linewidth}
		\centering
		\includegraphics[width=\linewidth]{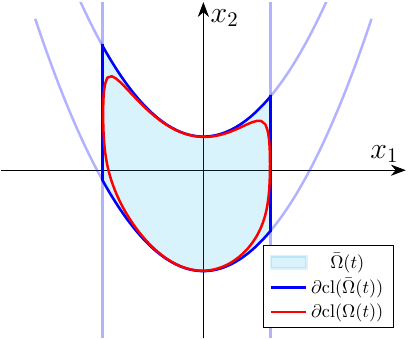}
		\caption{}
		\label{fig:condition II}
	\end{subfigure}%
	\caption{Snapshots of $\Ob(t)$ and its corresponding inner-approximation under \eqref{smooth_alph} for two different examples. \vspace{-0.4cm}}
	\label{fig:example1_cases}
\end{figure}

\begin{assumption} \label{assum:coercive_alphabar}
	For all $t\geq 0$, $-\alphb(t,x)$ is coercive (radially unbounded) in $x$, i.e, $-\alphb(t,x) \rightarrow +\infty$ as $\|x\| \rightarrow +\infty$. 
\end{assumption}

Note that while the focus of this work is on satisfying the output constraints defined in \eqref{eq:outputconst}, it is also required to design $u(t,x)$ such that the states of the closed-loop system \eqref{eq:sys_dynamics_firstorder} remain bounded for all times. Therefore, it is essential to ensure that the output-constrained set $\Ob(t)$ is bounded, which, in turn, guarantees the boundedness of $\|x\|$ for all times. The above assumption provides a necessary and sufficient condition for the boundedness of $\Ob(t)$ (resp. $\Omega(t)$) $\forall t \geq 0$, which is established in the following lemma:
\begin{lemma} \label{lem:Omegab_bounded}
	Under Assumption \ref{assum:coercive_alphabar}, $\Ob(t)$ (resp. $\Omega(t)$) is a bounded set for all $t \geq 0$. 
\end{lemma}
\begin{proof}
	See Appendix \ref{appen:proof_lemma_Omegab_bounded}.
\end{proof}
Note that, in Assumption \ref{assum:coercive_alphabar}, $-\alphb(t,x)$ should approach $+\infty$ along \textit{any path} within $\R^n$ on which $\|x\|$ tends to infinity. Define $h_f(x) \coloneqq \col(h_i) \in \R^p, i = \{1,\ldots, p\}$, $h_{\mathrm{L}}(x) \coloneqq \col(h_i) \in \R^q, i = \{p+1,\ldots, p+q\}$, and $h_{\mathrm{U}}(x) \coloneqq \col(h_i) \in \R^{m-p-q}, i = \{p+q+1,\ldots, m\}$, as the stacked vectors of system outputs associated with funnel, LBO, and UBO constraints in \eqref{eq:outputconst}, respectively. The following lemma provides explicit conditions on $h_i(x), i \in \I$, which ensure $-\alphb(t,x)$ (resp. $-\alpha(t,x)$) to be coercive. 

\begin{lemma}\label{lem:alphb_coercive_h}
	 The function $-\alphb(t,x)$ (resp. $-\alpha(t,x)$) is coercive in $x$ for all $t\geq 0$ if and only if at least one of the following conditions holds: (I): $\|h_f(x)\| \rightarrow +\infty$, (II): one or more elements of $h_{\mathrm{L}}(x)$ approaches $-\infty$, (III) one or more elements of $h_{\mathrm{U}}(x)$ approaches $+\infty$, along any path in $\R^n$ as $\|x\| \rightarrow +\infty$.       
\end{lemma}
\begin{proof}
	See Appendix \ref{appen:proof_lemma_alphb_coercive_h}.
\end{proof}

It is important to emphasize that the boundedness of $\Ob(t)$ is a technical requirement in this paper. One straightforward approach to guarantee this is by introducing a suitable auxiliary UBO constraint: $h_{\mathrm{aux}}(x) \coloneqq |x| < c_{\mathrm{aux}}$, where $c_{\mathrm{aux}}>0$ is a sufficiently large constant. This auxiliary constraint defines a large ball centered at the origin that encompasses all the other time-varying constraints in \eqref{eq:outputconst}. Having this auxiliary constraint along other constraints in \eqref{eq:outputconst} ensures the satisfaction of Lemma \ref{lem:alphb_coercive_h}'s condition.

\begin{remark} \label{rem:interpret_Omegab_bounded_in_h(x)}
	In Example \ref{ex:concave_example}, we have $h_f(x) = h_1(x)$, $h_L(x) = h_2(x)$, $h_U(x) = h_3(x)$, and one can verify that the condition of Lemma \ref{lem:alphb_coercive_h} is always satisfied along any path within $\R^2$ on which $\|x\| \rightarrow +\infty$. Therefore, $-\alphb(t,x)$ (and also $\alpha(t,x)$) is coercive and thus $\Ob(t)$ (resp. $\Omega(t)$) is bounded according to Lemma \ref{lem:Omegab_bounded} (see Fig.\ref{fig:omega_y_bar_and_Omega_y}). However, in this example, if we drop the LBO constraint $\rl_2(t) < h_2(x) = x_2 - x_1$ then $\Ob(t)$ (resp. $\Omega(t)$) will not be bounded anymore since along the path on which $x_1 = 0$ and $x_2 \rightarrow -\infty$, we get $h_f(x) = 0$, $h_U(x) \rightarrow -\infty$. Similarly, the boundedness of $\Ob(t)$ (resp. $\Omega(t)$) in Example \ref{ex:independent_funnels_example} can be verified.   
\end{remark}

Note that Assumption \ref{assum:coercive_alphabar} also guarantees existence of at least one global maximum for $\alphb(t,x)$ (resp. $\alpha(t,x)$) $\forall t \geq 0$ \cite[p. 27]{peressini1988mathematics}. In this regard, for each time instant $t$ we define:
\begin{equation} \label{eq:alphab_opt}
	\alphbopt(t) \coloneqq \max_{x \in \R^n} \alphb(t,x),
\end{equation} 
where $\alphbopt(t)$ denotes the maximum value of $\alphb(t,x)$ at time $t$. It is clear that if $\alphbopt(t) > 0$ the time-varying output constraints are feasible at time $t$, whereas $\alphbopt(t) \leq 0$ indicates that the constraints are infeasible at time $t$, thus impossible to be satisfied. Similarly, for a given $\nu$ in \eqref{smooth_alph} we can define:
\begin{equation}\label{eq:alpha_opt}
	\alphaopt(t) \coloneqq \max_{x \in \R^n} \alpha(t,x) \leq \alphbopt(t).
\end{equation}
From \eqref{eq:alpha_opt} and \eqref{smooth_alph} one can claim that having $\alphaopt(t) > 0$ is \textit{sufficient} for the feasibility of the time-varying output constraints \eqref{eq:outputconst} at time $t$. In this paper, for simplicity in the control design, we let the time-varying output constraints in \eqref{eq:outputconst} are mutually satisfiable for all times, which is summarized in the following assumption:

\begin{assumption}\label{assu:feasible_output_constr}
	There exists $\epsilon_{\alpha}>0$ such that $\alphaopt(t) \geq \epsilon_{\alpha}  > 0, \forall t\geq 0$, i.e., $\Omega(t)$ is non-empty (feasible) for all times.
\end{assumption}

\subsection{Designing A Single Funnel Constraint}
\label{subsec:single_funnel_const}

As mentioned in Subsection \ref{subsec:scalar_metric}, the satisfaction of \eqref{eq:outputconst} can be ensured by keeping $\alpha(t,x(t;x_0))$ positive. Therefore, the control design problem in Section \ref{sec:problemformulation} boils down to designing $u(t,x)$ for \eqref{eq:sys_dynamics_firstorder} such that if $\alpha(0,x_0) > 0$ then $\alpha(t,x(t;x_0)) > 0, \forall t\geq 0$, and if $\alpha(0,x_0) < 0$ then $\alpha(t,x(t;x_0)) > 0, \forall t\geq T$. To achieve this one can design a funnel-based control law to ensure the following \textit{single funnel constraint} for \eqref{eq:sys_dynamics_firstorder}: 
\begin{equation} \label{eq:alpha_funnelconst}
	\rlalph(t) < \alpha(t,x(t;x_0)) < \rualph(t),
\end{equation}
where $\rl_{\alpha}, \rualph: \Rpos \rightarrow \R$ are properly designed bounded and continuous functions of time with bounded derivatives. To avoid any ambiguity between the funnel constraint in \eqref{eq:alpha_funnelconst} and output funnel constraints in \eqref{eq:outputconst}, we refer to \eqref{eq:alpha_funnelconst} as \textit{$\alpha$-funnel constraint}. Note that \eqref{eq:alpha_funnelconst} is feasible (valid) when: \textbf{(i)} $\rualph(t) - \rlalph(t) \geq \delta_{\rho}, \forall t\geq 0$, for some $\delta_{\rho} > 0$, and \textbf{(ii)} $\rlalph(t) < \alphaopt(t), \forall t \geq 0$, since $\alpha(t,x) \leq \alphaopt(t), \forall t \geq 0$. 

In virtue of Assumption \ref{assu:feasible_output_constr}, i.e., $\alphaopt(t) > 0, \forall t \geq 0$, one can design $\rualph(t)$ such that $\rualph(t) - \alphaopt(t) \geq \bar{\varsigma} > 0, \forall t \geq 0$, and accordingly design $\rlalph(t)$ such that $\alphaopt(t) - \rlalph(t) \geq \underline{\varsigma} > 0, \forall t \geq 0$, to ensure conditions (i) and (ii) for the feasibility of \eqref{eq:alpha_funnelconst}. In particular, one can set $\rualph(t) = \rho_{\max}$, where $\rho_{\max}$ is a sufficiently large positive constant, such that $\rho_{\max} > \sup (\alphaopt(t))$. Then, for ensuring the satisfaction of \eqref{eq:outputconst} through imposing \eqref{eq:alpha_funnelconst}, one simple strategy for designing $\rlalph(t)$ is as follows: \textbf{(a)} if $\alpha(0,x_0) > 0$ (i.e., the constraints are initially satisfied), set $\rlalph(t) = 0$, and \textbf{(b)} if $\alpha(0,x_0) < 0$, design $\rlalph(t)$ such that $\rlalph(0) < \alpha(0,x_0) < 0$ and $\rl_{\alpha}(t\geq T) = 0$, i.e., the lower bound in \eqref{eq:alpha_funnelconst} increases with time so that it enforces $\alpha(t,x(t;x_0))$ to become positive for $t \geq T > 0$. In this respect, inspired by \cite{yin2020robust}, given a desirable $T\geq 0$, one can design:
\begin{equation}\label{eq:alpha_lower_bound}
	\rl_{\alpha}(t) = \begin{cases}
							\left( \frac{T-t}{T} \right)^{\frac{1}{1-\beta}} (\rho_0 - \rho_{\infty}) + \rho_{\infty}, & 	0 \leq t < T,\\
							\rho_{\infty}, & t \geq T,
					  \end{cases}
\end{equation}
where $\beta \in (0,1)$ is a constant. Note that $\rlalph(0) = \rho_0$ and $\rl_{\alpha}(t \geq T) = \rho_{\infty}$. Therefore, for designing $\rl_{\alpha}(t)$, for case (a) we set $\rho_0 = \rho_{\infty} = 0$ and for case (b) we set $\rho_0$ such that $\rlalph(0) = \rho_0 < \alpha(0,x_0) < 0$ and $\rho_{\infty} = 0$.

\begin{remark} \label{rem:compute_alpha_opt}
We highlight that, when Assumption \ref{assu:feasible_output_constr} holds, by taking a sufficiently large $\rho_{\max} > 0$ and setting $\rho_{\infty}=0$, the solution of the time-varying optimization problem \eqref{eq:alphab_opt} is not required for designing $\rualph(t)$ and $\rlalph(t)$ in \eqref{eq:alpha_funnelconst}. However, for taking $\rho_{\infty} > 0$ we require $\rho_{\infty} < \inf(\alphaopt(t))$ to hold for ensuring the feasibility of \eqref{eq:alpha_funnelconst}. Furthermore, given the readily available initial condition of the system \eqref{eq:sys_dynamics_firstorder}, ensuring $\rlalph(0) = \rho_0 < \alpha(0, x_0)$ is not restrictive.   
\end{remark}

\begin{remark} \label{rem:robustness_dgree}
	Under Assumption \ref{assu:feasible_output_constr}, choosing a larger $\rho_{\infty}$, under the condition that $0< \rho_{\infty} < \inf(\alphaopt(t))$, affects how well the time-varying output constraints should be satisfied for $t\geq T$, thus it constitutes a user-defined margin for constraints satisfaction.
\end{remark}

\subsection{Controller Design and Stability Analysis}
\label{subsec:control_design}

Now inspired by \cite{bechlioulis2014low}  we design a low-complexity model-free robust funnel controller using the PPC method for \eqref{eq:sys_dynamics_firstorder} to ensure the satisfaction of \eqref{eq:alpha_funnelconst}. In this respect, first, we define the normalized $\alpha(t,x)$ (w.r.t. the asymmetric funnels given by \eqref{eq:alpha_funnelconst}) as follows:
\begin{equation} \label{eq:normal_alpha}
	\alphah(t,x) \coloneqq \frac{\alpha(t,x) - 0.5 \, \rhos(t)}{0.5 \, \rhod(t)}, 
\end{equation}
where $\rhos(t) \coloneqq \rualph(t) + \rlalph(t)$ and $\rhod(t) \coloneqq \rualph(t) - \rlalph(t)$. Notice that $\alphah(t,x) \in (-1,1)$ if and only if $\alpha(t,x) \in (\rlalph(t), \rualph(t))$. Next, we introduce the following nonlinear transformation:
\begin{equation} \label{eq:mapped_alphah}
	\epialph(t,x) = \T(\alphah(t,x)) \coloneqq \ln \left( \frac{1+\alphah(t,x)}{1 - \alphah(t,x)} \right),
\end{equation}
where $\epialph$ is the unconstrained transformed signal corresponding to $\alpha(t,x)$ and $\T: (-1 , 1) \rightarrow (-\infty, +\infty)$ is a smooth strictly increasing bijective mapping, which satisfies $\T(0) = 0$. Note that maintaining boundedness of $\epialph$ enforces $\alphah(t,x) \in (-1,1)$, and thus the satisfaction of \eqref{eq:alpha_funnelconst}. 

To design the control law we proceed as follows: first, define $V(\epialph) = 0.5 \epialph^2$, which is a positive definite (implicitly time-varying) \textit{barrier function} associated with the $\alpha$-funnel constraint in \eqref{eq:alpha_funnelconst}. Note that $V(0) = 0$ and as $\alpha(t,x)$ approaches $\rlalph(t)$ or $\rualph(t)$ (i.e., as $\alphah(t,x)$ approaches $\pm1$) we get $V(\epialph) \rightarrow +\infty$. Next, from \eqref{eq:mapped_alphah}, with a slight abuse of notation, one can consider ${V}(t, x)$ and design a (gradient-based) control law as follows:
\begin{equation} \label{eq:PPC_control}
	u(t,x) \coloneqq - k \, \nabla_x V(t,x),
\end{equation}
where $k > 0$ is the control gain and $\nabla_x$ denotes the gradient with respect to $x$. Applying the chain rule in \eqref{eq:PPC_control} gives $u(t,x)$ more explicitly as:
\begin{flalign} \label{eq:PPC_control_explicit}
	&u(t,x) = - k \, \left( \dfrac{\partial V(\epialph)}{\partial \epialph} \, \dfrac{\partial \epialph(\alphah)}{\partial \alphah} \,  \dfrac{\partial \alphah(t,\alpha)}{\partial \alpha} \,  \dfrac{\partial \alpha(t,x)}{\partial x} \right)^{\top}& \nonumber \\
	& \quad =  - k   \left(\dfrac{\partial \alpha(t,x)}{\partial x} \right)^{\top} \!\! \xialph \, \epialph = -k \, \nabla_x \alpha(t,x) \, \xialph \, \epialph,&
\end{flalign}
where $\xialph \coloneqq \frac{\partial \epialph(\alphah)}{\partial \alphah} \, \frac{\partial \alphah(t,\alpha)}{\partial \alpha}$ is given by:
\begin{equation}\label{eq:xi}
	\xialph(t,x) = \dfrac{4}{ \rhod(t) ( 1- \alphah^2(t,x) )}. 
\end{equation} 

Recall that \eqref{eq:alpha_funnelconst} is ensured through keeping $\epialph$ bounded, which is equivalent to establishing boundedness of $V(\epialph)$ through applying \eqref{eq:PPC_control_explicit} in \eqref{eq:sys_dynamics_firstorder}. 

We require the following assumption on $\alpha(t,x)$ to avoid $\nabla_x \alpha(t,x) = \mathbf{0}_n$ at certain undesired singular points, which can lead to controllability loss issues when \eqref{eq:PPC_control_explicit} is applied to \eqref{eq:sys_dynamics_firstorder}.

\begin{assumption}\label{assu:alpha_globalmax}
	For all $t \geq 0$ the function $-\alpha(t,x)$ is invex, i.e., every critical point of $\alpha(t,x)$ is a (time-varying) global maximizer (see \cite[Theorem 2.2]{mishra2008invexity}). 
\end{assumption}

The following lemma gives some \textit{sufficient} conditions for ensuring Assumption \ref{assu:alpha_globalmax}.

\begin{lemma} \label{lem:global_max_suffi}
	The function $-\alpha(t,x)$ is invex $\forall t \geq 0$ if any of the following conditions holds: \\
	\textbf{I}. $\psi_i(t,x), \forall i \in \I_{\psi}$ in \eqref{eq:psi_constr} are concave in $x$ for all $t\geq 0$. \\
	\textbf{II}. All of the $m$ time-varying output constraints in \eqref{eq:outputconst} are funnel constraints such that: (i) $n = m = p$, (ii) the output map $y = h(x)$ in \eqref{eq:sys_dynamics_firstorder} is norm-coercive (i.e., $\|h(x)\| \rightarrow +\infty \quad \text{as} \quad \|x\| \rightarrow +\infty$), and (iii) the Jacobian matrix $J(x) \coloneqq \frac{\partial h(x)}{\partial x} \in \R^{n \times n}$ is full rank for all $x \in \R^{n}$.
\end{lemma}
\begin{proof}
	See Appendix \ref{appen:proof_lemma_globalmax}.
\end{proof}

\begin{remark}\label{rem:interpret_Lemm_globalmax_in_h}
	The concavity of $\psi_{i}(t,x), \forall t \geq 0$ in Lemma \ref{lem:global_max_suffi} can be understood by examining \eqref{eq:predicate_const_rep} in terms of $h_i(x), i \in \I$. Specifically, for funnel constraints, the functions $h_i(x), i=\{1,\ldots,p\}$ should be affine (linear) functions as $\psi_{2i}(t,x)$ and $\psi_{2i-1}(t,x), i=\{1,\ldots,p\}$ are concave only when $h_i(x)$ and $-h_i(x)$ are concave in \eqref{eq:predicate_funnel_rep}. Moreover, from \eqref{eq:predicate_onesided_rep} for LBO constraints $h_i(x), i = \{p+1,\ldots p+q\}$ should be concave, while for UBO constraints $h_i(x), i = \{p+q+1,\ldots m\}$ need to be convex. In this regard, the interior of \textit{any time-varying bounded convex polytope} in $\R^n$ can be considered as $\Ob(t)$, for which its corresponding $\alpha(t,x)$ satisfies Assumption \ref{assu:alpha_globalmax}. This is because for a bounded convex polytope all $h_i(x), i\in\I$ should be affine (thus condition I of Lemma \ref{lem:global_max_suffi} holds) and also Assumption \ref{assum:coercive_alphabar} is readily satisfied.
\end{remark}

\begin{remark}
	Followed by Remark \ref{rem:interpret_Lemm_globalmax_in_h} one can verify that Example \ref{ex:concave_example} (depicted in Fig. \ref{fig:omega_y_bar_and_Omega_y}) satisfies condition I of Lemma \ref{lem:global_max_suffi}. Moreover, Example \ref{ex:independent_funnels_example} (depicted in Fig.\ref{fig:condition II}) satisfies condition II of Lemma \ref{lem:global_max_suffi} since in this example $n=m=p=2$, the Jacobian matrix of $h(x)$, $J(x) = \left[\begin{smallmatrix} 1 & 0 \\ 0.6x_1 & -1 \end{smallmatrix} \right]$ is full rank for all $x\in\R^2$, and $h(x)=h_f(x)$ is norm-coercive. Note that Example \ref{ex:independent_funnels_example} does not satisfy condition I of Lemma \ref{lem:global_max_suffi} as $h_2(x)$ in is not affine. It is worth emphasizing that condition II in Lemma \ref{lem:global_max_suffi} accurately captures the notion of independence between $n$ funnel constraints in $\R^n$. This means that the satisfaction of individual feasible funnel constraints does not interfere with each other, meaning that the funnel constraints are decoupled. 
\end{remark}

\begin{remark}
	Note that condition I of Lemma \ref{lem:global_max_suffi} is not enough for ensuring the boundedness of $\Omega(t)$. Indeed, to guarantee the boundedness of $\Omega(t)$, $\psi_{i}(t,x), i \in \I_{\psi}$ should be such that the condition of Lemma \ref{lem:alphb_coercive_h} is also met. However, regarding condition II of Lemma \ref{lem:global_max_suffi}, since $h(x)$ is norm-coercive and only funnel type constraints are considered, i.e., $h(x) = h_f(x)$, one can verify that condition I in Lemma \ref{lem:alphb_coercive_h} is already satisfied, which ensures the boundedness of $\Omega(t)$.
\end{remark}

The following theorem summarizes our main result:
\begin{theorem} \label{th:main}
	Consider the nonlinear input affine system \eqref{eq:sys_dynamics_firstorder} with the time-varying output constraints \eqref{eq:outputconst}. Let $\rlalph(t), \rualph$ be designed based on the discussion in Section \ref{subsec:single_funnel_const}, such that $\rlalph(0) < \alpha(0,x_0) < \rualph(0)$. Under Assumptions 1-6 the feedback control law \eqref{eq:PPC_control_explicit} guarantees the satisfaction of the $\alpha$-funnel constraint \eqref{eq:alpha_funnelconst} for all times, as well as the boundedness of all closed-loop signals.
\end{theorem}
\begin{proof}
	See Appendix \ref{appen:proof_theorem}.
\end{proof}

\begin{remark}
	We underline that the controller developed in \eqref{eq:PPC_control_explicit} adeptly manages coupled time-varying output constraints, such as those exemplified in Example \ref{ex:concave_example}. This sets it apart from previous controller designs reliant on methods like Time-varying Barrier Lyapunov Functions, Funnel Control, and Prescribed Performance Control \cite{tee2011control, ilchmann2002tracking, bechlioulis2008robust}, which often encounter challenges when dealing with such intricate scenarios. Furthermore, unlike these approaches and the work in \cite{yin2020robust}, our method does not necessitate the initial satisfaction of all output constraints and provides the added advantage of customizable finite-time constraint satisfaction. Furthermore, this work builds upon the methodology presented in \cite{lindemann2017prescribed, lindemann2021funnel}, which was initially designed to manage STL specifications by transitioning between time-invariant constrained sets. However, our work takes a step further by extending the proposed methodology to effectively tackle scenarios involving time-varying constrained sets.
\end{remark}

\section{Simulation Results}
\label{sec:simu_results}

Consider an unstable dynamical system \eqref{eq:sys_dynamics_firstorder}, where 
\begin{equation*}
	\resizebox{1\hsize}{!}{$ f(x) = \begin{bmatrix}
	 -x_1^2 x_2 - x_1^3 - e^{-x_1^2-x_2^2}  \\
	 0.1 x_2^2  + x_1^2 + \sin(x_1 x_2) 
	\end{bmatrix}, \;
	g(x) = \begin{bmatrix}
		x_2^2 + 1 &  \cos(x_1) \\
		\sin(x_2) & x_1^2 + 2
	\end{bmatrix},$}
\end{equation*}
and $w(t) = [w_1(t), w_2(t)]^\top \in \R^2$, in which $w_1(t) = 1.5 \sin(2t + \frac{\pi}{3}) + 3 \cos(3t + \frac{3\pi}{7})$ and $w_2(t) = 0.5 \sin(3t) e^{\cos(2t+\frac{\pi}{3})+ 1}$. Moreover, let the output map $h(x)$ and the corresponding time-varying output constraints for \eqref{eq:sys_dynamics_firstorder} be the same as Example \ref{ex:concave_example} with $\ru_1(t) = 3\sin(0.3t)$, $\rl_1(t) = - 2 + 2.5\sin(0.3t)$, $\rl_2(t) = - \cos(0.3t)$, and $\ru_3(t) = 3.5 - \cos(0.3t)$. The considered output constraints ensure $\Omega(t)$ to be feasible at all times. Indeed, solution of \eqref{eq:alpha_opt} for a sufficiently long time horizon reveals that $\sup(\alphaopt(t)) < 1.1 $, and $\inf(\alphaopt(t)) > 0.3$ (see Fig. \ref{fig:ex1_alpha_evolue_x_evolu}, top left), hence, Assumption \ref{assu:feasible_output_constr} holds. The initial condition of the system is assumed to be $x_0 = [2, -2.5]^\top$, which gives $\alpha(0,x_0) < 0$. Followed by the discussion in Subsection \ref{subsec:single_funnel_const} and having the knowledge on $\sup(\alphaopt(t))$ and $\inf(\alphaopt(t))$ we can take $\rualph(t) = 1.5$ and set $\rho_{\infty} = 0.1$ in \eqref{eq:alpha_lower_bound}. If we took for granted the validity of Assumption \ref{assu:feasible_output_constr} in this simulation example, from Remark \ref{rem:compute_alpha_opt} we know that the control law \eqref{eq:PPC_control_explicit} does not necessarily require any information on $\alphaopt(t)$, e.g., we could set, $\rho_{\infty} = 0$ and $\rualph(t) = 50$ (a sufficiently large constant) without any prior knowledge on $\sup(\alphaopt(t))$ and $\inf(\alphaopt(t))$. Note that $\rho_0$ in \eqref{eq:alpha_lower_bound} is selected such that $\rlalph(0) < \alpha(0,x_0)<0$. Other parameters are set as follows: $\beta = 0.5$ and $T = 6$ for $\rlalph(t)$ in \eqref{eq:alpha_lower_bound}, and $k = 1$ and $\nu = 10$ in \eqref{eq:PPC_control_explicit} and \eqref{smooth_alph}, respectively. Fig.\ref{fig:ex1_alpha_evolue_x_evolu} (top left) depicts the evolution of $\alpha(t,x(t;x_0))$ under the proposed control law \eqref{eq:PPC_control_explicit}. Since $\alpha(0,x_0) < 0$, it is clear that the output constraints are not initially satisfied, however, by imposing the $\alpha$-funnel constraint \eqref{eq:alpha_funnelconst} $\alpha(t,x(t;x_0))$ becomes and then remains positive with a margin of $\rho_{\infty} = 0.1$. Snapshots of the evolution of $x(t;x_0)$ along with the time-varying constrained set $\Omega(t)$, for which $\partial\cl(\Omega(t)) = \{x \in \R^2 \mid \alpha(t,x) = 0\}$, are also illustrated in Fig.\ref{fig:ex1_alpha_evolue_x_evolu}.

\begin{figure}[tbp]
	\centering
	\begin{subfigure}[t]{0.57\linewidth}
		\includegraphics[width=\linewidth]{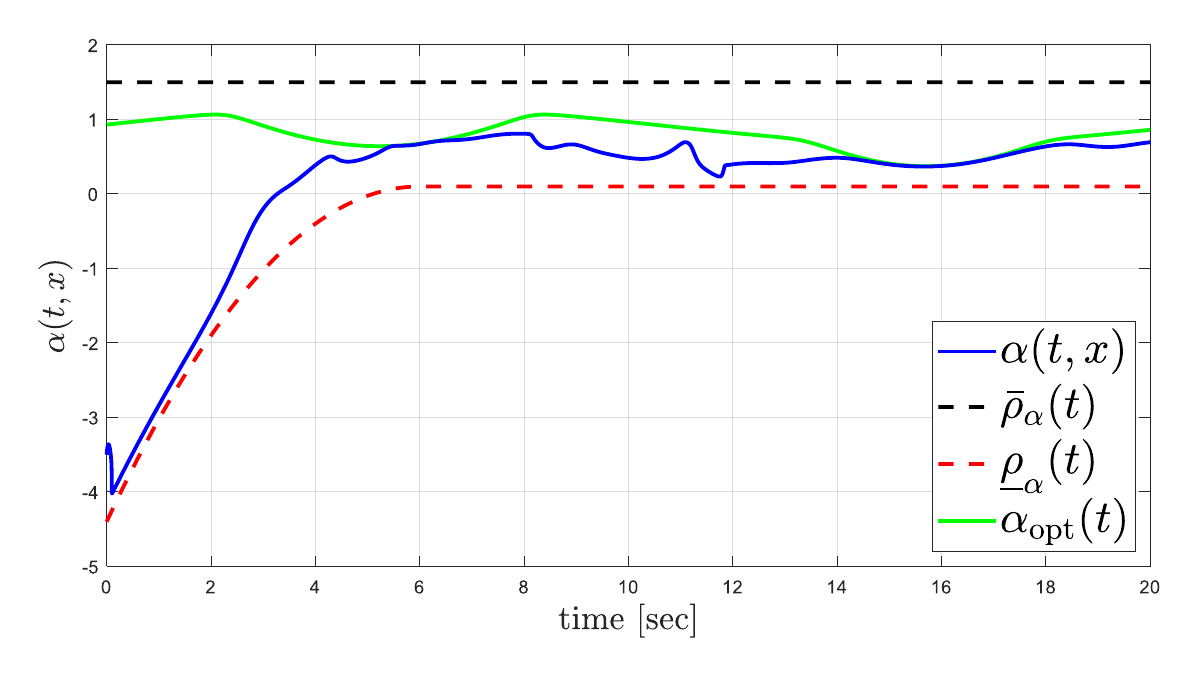}
	\end{subfigure}
	~ \hspace{0.55cm}
	\begin{subfigure}[t]{0.3\linewidth}
		\includegraphics[width=\linewidth]{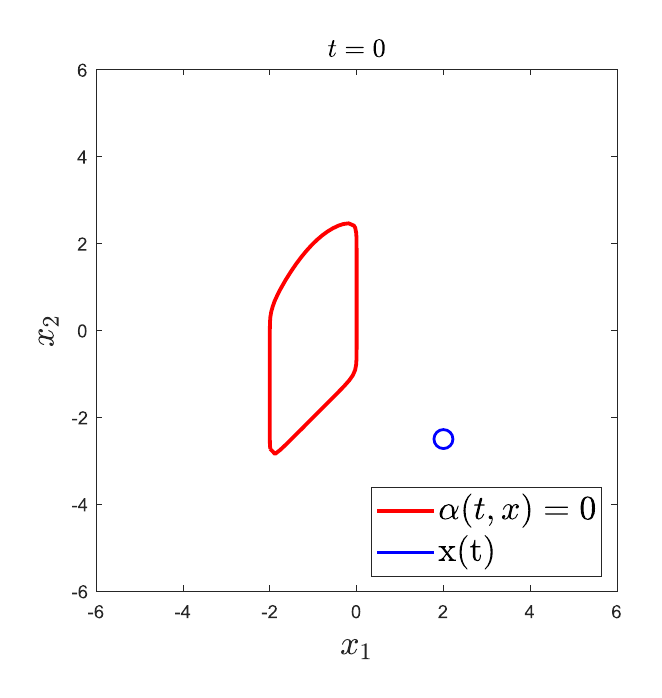}
	\end{subfigure}

	\begin{subfigure}[t]{0.3\linewidth}
		\includegraphics[width=\linewidth]{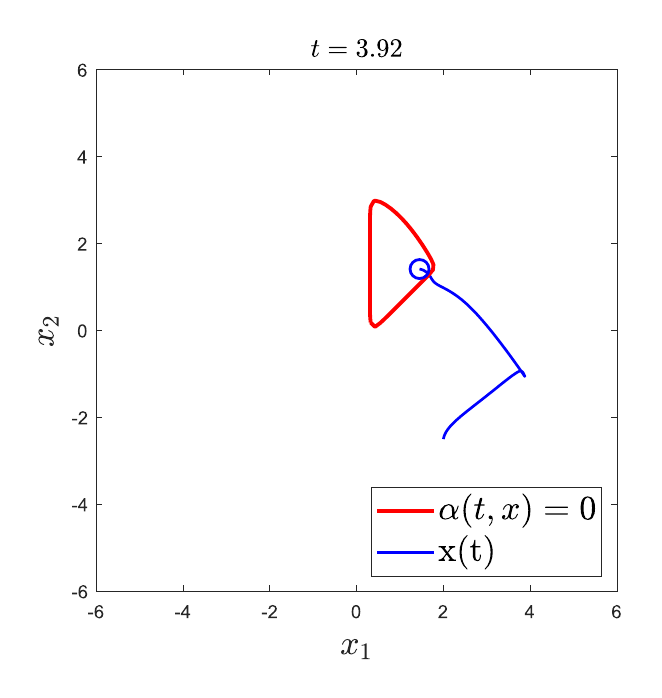}
	\end{subfigure}
	~
	\begin{subfigure}[t]{0.3\linewidth}
		\includegraphics[width=\linewidth]{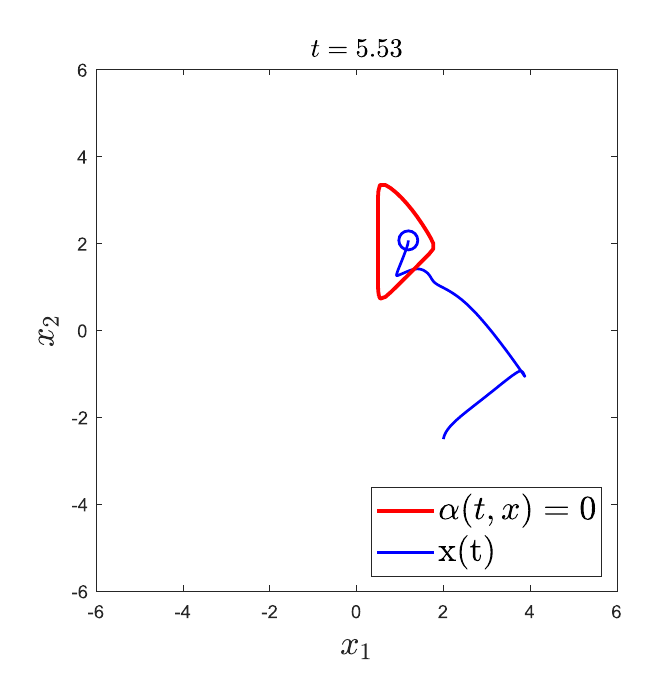}
	\end{subfigure}
	~
	\begin{subfigure}[t]{0.3\linewidth}
		\includegraphics[width=\linewidth]{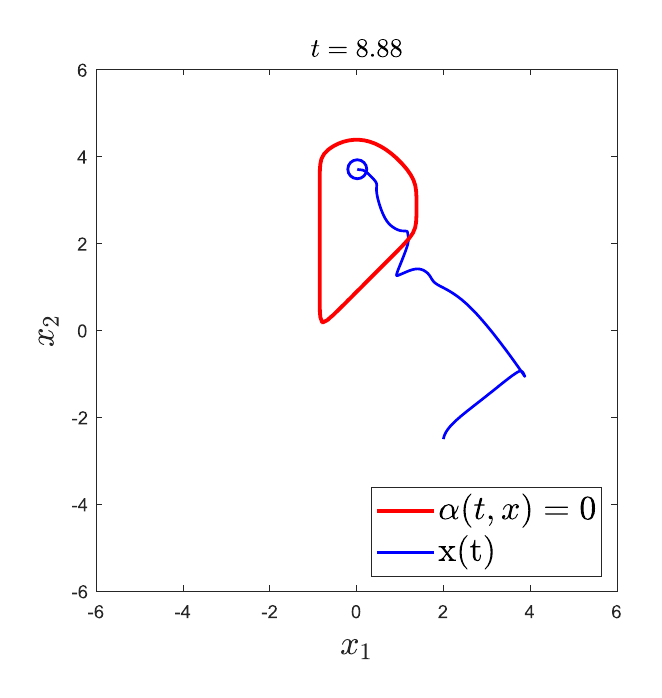}
	\end{subfigure}
	~
	\begin{subfigure}[t]{0.3\linewidth}
		\includegraphics[width=\linewidth]{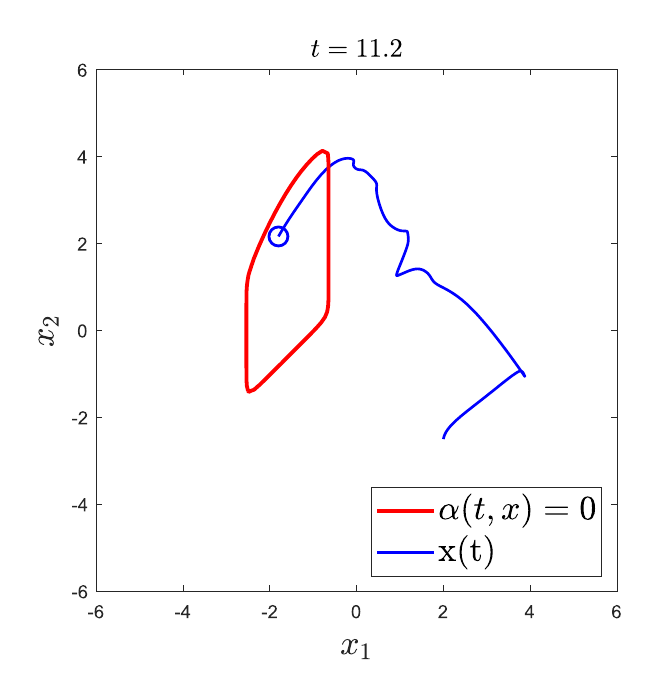}
	\end{subfigure}
	~
	\begin{subfigure}[t]{0.3\linewidth}
		\includegraphics[width=\linewidth]{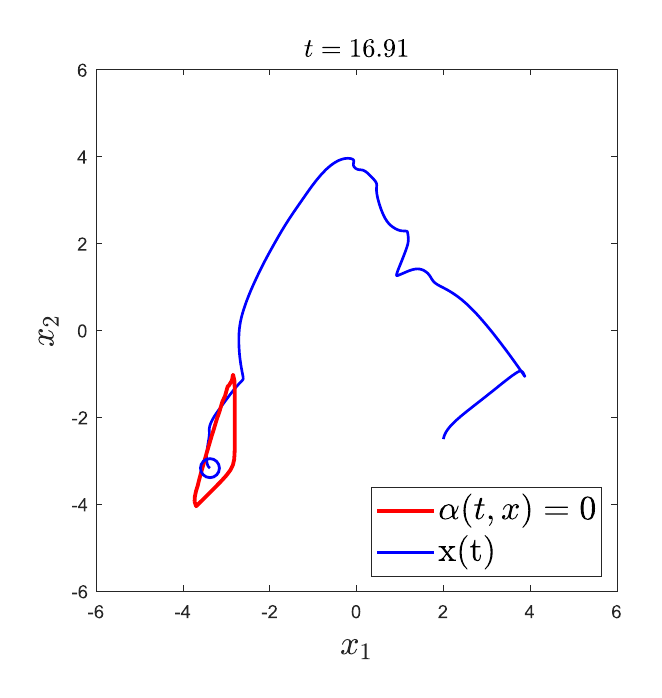}
	\end{subfigure}
	~
	\begin{subfigure}[t]{0.3\linewidth}
		\includegraphics[width=\linewidth]{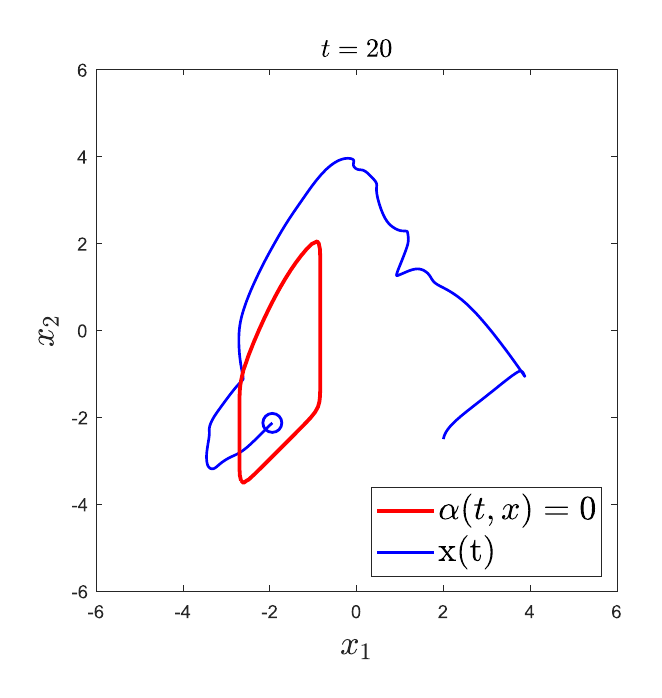}
	\end{subfigure}
	\caption{Evolution of $\alpha(t,x(t;x_0))$ (top-left) and $x(t)$ under \eqref{eq:PPC_control_explicit}.}
	\label{fig:ex1_alpha_evolue_x_evolu}
\end{figure}

\section{Conclusions}
\label{sec:conclu}

We considered the control design problem for nonlinear uncertain systems under (potentially coupled) multiple time-varying output constraints. The proposed method decodes the constraint satisfaction for multiple outputs as the satisfaction of a single funnel constraint on the signed distance w.r.t. the boundary of the time-varying output-constrained set. A robust low-complexity controller is designed based on the prescribed performance control approach to ensure the satisfaction of the single funnel constraint for the closed-loop system dynamics. The developed results in this paper are also valid for time-invariant output constraints as a special case. Future work will be devoted to the extension of the proposed control method to higher-order nonlinear systems as well as relaxing Assumption \ref{assu:feasible_output_constr}.

\bibliographystyle{ieeetr}
\bibliography{Refs}

\appendices

\section{Proof of Lemma \ref{lem:Omegab_bounded}}
\label{appen:proof_lemma_Omegab_bounded}

First, note that since all $\ru_i(t), \rl_i(t)$ in \eqref{eq:predicate_const_rep} are bounded, $\psi_i(t,x), i \in \I_{\psi}$ are bounded $\forall t\geq 0$ for any fixed $x$. Hence, the scalar continuous function $\alphb: \Rpos \times \R^n \rightarrow \R$ in \eqref{eq:metric} is bounded $\forall t\geq 0$ for any fixed $x$. Define $\alphb_{t}(x) \coloneqq \alphb(t,x)$. Owing to Assumption \ref{assum:coercive_alphabar}, $-\alphb_{t}(x)$ is coercive in $x$ for each $t$, therefore, from \cite[Proposition 11.12]{bauschke2011convex} all super-level sets $\alphb_{t}(x) \geq c$, where $c \in \R$, are bounded at each time instant.  Furthermore, based on Assumption \ref{assum:coercive_alphabar} and \cite[Theorem 1.4.4, p. 27]{peressini1988mathematics}, we can infer that there exists a time-dependent constant $\bar{c}(t) \in \R$ such that the super-level sets $\alphb_{t}(x_1) > \bar{c}(t)$ are empty. This implies that $\Ob(t)$ in \eqref{eq:omega_alpha_bar} is bounded, and in particular, it is empty if $\bar{c}(t) < 0$ at time $t$. Moreover, from $\alphb(t,x) \leq \alpha(t,x) + \frac{1}{\nu} \ln(m+p)$ we can verify that $-\alphb(t,x_1)$ is coercive if and only if $-\alpha(t,x_1)$ is coercive. As a result, we can apply similar arguments as mentioned above to establish the boundedness of $\Omega(t)$.

\section{Proof of Lemma \ref{lem:alphb_coercive_h}}
\label{appen:proof_lemma_alphb_coercive_h}

Using \eqref{eq:metric}, we can determine whether $-\alphb(t,x)$ is coercive is coercive if and only if \textit{at least} one of the functions $\psi_i(t,x)$ in \eqref{eq:predicate_const_rep} approaches $-\infty$ as $\|x\| \rightarrow +\infty$ (along any path on $\R^n$). It's worth noting that $-\alpha(t,x)$ in \eqref{smooth_alph} is also coercive under the same condition. Since the functions $\psi_i(t,x), i \in \I_{\psi}$ in \eqref{eq:predicate_const_rep} are bounded for all $t\geq 0$ and any fixed $x$ (see proof of Lemma \ref{lem:Omegab_bounded}), we can interpret this requirement in terms of $h_i(t,x)$. Specifically, if there exists an $i \in \{1,\ldots,p\}$ such that $h_i(t,x) \rightarrow \pm \infty$, then it ensures that $\psi_{2i-1}(t,x) \rightarrow -\infty$ or $\psi_{2i}(t,x) \rightarrow -\infty$ in \eqref{eq:predicate_funnel_rep} and vice versa. To simplify the verification process, we only need to check whether $\|h_f(t,x)\| \rightarrow +\infty$ holds along a path on $\R^n$ as $\|x\| \rightarrow +\infty$. From \eqref{eq:predicate_onesided_rep}, we can also see that if there exists a $j \in \{p+1,\ldots,p+q\}$ such that $h_j(t,x) \rightarrow - \infty$, then there exists an $i \in \{2p+1,\ldots,2p+q\}$ such that $\psi_{i}(t,x) \rightarrow - \infty$ and vice versa. Similarly, if there exists a $j \in \{p+q+1,\ldots,m\}$ such that $h_j(t,x) \rightarrow + \infty$, then there exists an $i \in \{2p+q+1,\ldots,p+m\}$ such that $\psi_{i}(t,x) \rightarrow - \infty$ and vice versa. In sum, if along any path on $\R^n$ as $\|x_1\| \rightarrow +\infty$ at least one of the above conditions holds (conditions I-III in the lemma) then $-\alphb(t,x)$ (resp. $-\alpha(t,x)$) is coercive and vice versa.

\section{Proof of Lemma \ref{lem:global_max_suffi}}
\label{appen:proof_lemma_globalmax}

\textbf{Case I:} Consider $\alpha(t,x)$ in \eqref{smooth_alph}. First, note that since $\psi_{i}(t,x), i \in \I_{\psi}$ are concave functions in $x \in \R^n, \forall t \geq 0$ then as $\nu > 0$, $-\nu \, \psi_{i}(t,x), i \in \I_{\psi}$ are convex $\forall t\geq 0$. Therefore, from \cite[Section 3.5]{boyd2004convex} it is known that $e^{-\nu \, \psi_{i}(t,x)}, i \in \I_{\psi}$ are log-convex functions. Hence, $\sum_{i=1}^{m+p} e^{- \nu  \, \psi_i(t,x)}$ is log-convex. Consequently, $\alpha(t,x)$ in \eqref{smooth_alph} is a concave function $\forall t \geq 0$, and since $\alpha(t,x)$ has bounded level sets (from Assumption \ref{assum:coercive_alphabar}) it attains a well-defined global maximum (i.e., the global maximum exists). Therefore, one can conclude that every critical point of $\alpha(t,x)$ is a global maximizer.

\textbf{Case II:} Here we first establish that under the given conditions $\alpha(t,x)$ attains only one critical point and then we show that the critical point is the (unique) global maximizer of $\alpha(t,x)$. Recall that the critical points of $\alpha(t,x)$ are obtained by solving $\nabla_x \alpha(t,x) = 0$. Given the assumed ordering of constraint types in \eqref{eq:predicate_const_rep} one can write $\alpha(t,x)$ in \eqref{smooth_alph} as follows:
\begin{align}\label{smooth_alph_alt}
	&\alpha(t,x) = -\dfrac{1}{\nu} \ln \Big( \sum_{i=1}^{p}  e^{- \nu  \, ( h_i(x) - \rl_i(t) )} + e^{- \nu  \, (\ru_i(t) - h_i(x))} \nonumber \\
	& + \!\! \sum_{i=p+1}^{p+q} \! e^{- \nu  \, ( h_i(x) - \rl_i(t) )} + \!\!\! \sum_{i=p+q+1}^{m} \!\!\! e^{- \nu  \, (\ru_i(t) - h_i(x))} \Big).
\end{align}
Using \eqref{smooth_alph_alt} and \eqref{smooth_alph}, and after some calculations, we can obtain $\nabla_x \alpha(t,x)$ in a compact form as:
\begin{equation} \label{eq:grad_alph}
	\nabla_x \alpha(t,x) = \left( \dfrac{\partial \alpha(t,x)}{\partial x} \right)^{\top} = J^{\top}(x) \, \gamma(t,x) \, e^{\nu \alpha(t,x)},
\end{equation}
where $J = \frac{\partial h(x)}{\partial x} \in \R^{m \times n}$ is the Jacobian of $y = h(x)$, and $\gamma(t,x) \coloneqq  \col(\gamma_i(t,x)) \in \R^{m}$, in which $\gamma_i(t,x), i \in \I$ are given by: 
\begin{subnumcases}{\label{eq:gamma_vec}} 
	e^{- \nu  \, ( h_i(x) - \rl_i(t))} - e^{- \nu  \, \left(\ru_i(t) - h_i(x)\right)}, & \hspace{-0.3cm} $i = \{1,\ldots,p\}$ \label{eq:gamma_i_funnel} \\
	e^{- \nu  \, ( h_i(x) - \rl_i(t))},  & \hspace{-2.3cm} $i = \{p+1,\ldots,p+q\}$ \label{eq:gamma_i_lower} \\
	- e^{- \nu  \, \left(\ru_i(t) - h_i(x)\right)}, & \hspace{-2.3cm} $i = \{p+q+1,\ldots,m\}$ \label{eq:gamma_i_upper} 
\end{subnumcases}

Notice that in \eqref{eq:grad_alph}, $e^{\nu \alpha(t,x)} > 0$, therefore, $\nabla_x \alpha(t,x) = 0$ if and only if $J^{\top}(x) \, \gamma(t,x) = 0$. If $m = n$ and each output constraint is a funnel constraint (i.e., $p = m = n$) then all $\gamma_i(t,x)$ will be given by \eqref{eq:gamma_i_funnel}. In this case, $J(x) \in \R^{n \times n}$ in \eqref{eq:grad_alph} is a square matrix and if $\rank(J) = n$, $\forall x \in \R^n$, then $J^{\top}(x) \, \gamma(t,x) = 0$ holds if and only if $\gamma(t,x) = 0$. Therefore, under the above conditions, $\nabla_x \alpha(t,x) = 0$ holds if and only if $\gamma_i(t,x)  = 0, \forall i \in \I$, and owing to \eqref{eq:gamma_i_funnel} it leads to having $h_i(x) = 0.5 (\ru_i(t) + \rl_i(t)), \forall i \in \I$, hence, we get the following system of nonlinear equations:
\begin{equation} \label{eq:nonlin_equation}
	F(t,x) \coloneqq h(x) - 0.5\left( \ru(t) + \rl(t) \right) = 0,
\end{equation}
where $\ru(t) \coloneqq \col(\ru_i(t)) \in \R^{n}$, $\rl(t) \coloneqq \col(\rl_i(t)) \in \R^{n}$, and $F(x,t)$ is $\C^1$ in $t$ and $\C^2$ in $x$ owing to the properties of $h(x)$, $\ru(t)$, and $\rl(t)$. For each $t$ define $F_t(x) \coloneqq F(t,x)$ and note that $0.5\left( \ru(t) + \rl(t) \right)$ in \eqref{eq:nonlin_equation} is bounded for all times. We are interested in checking the existence and uniqueness of the solution to \eqref{eq:nonlin_equation} at each time instant $t$, which boils down to checking the existence and uniqueness of the solution to $F_t(x) = 0$ for each $t$. Since $h(x)$ is norm-coercive (i.e., $\|h(x)\| \rightarrow +\infty \quad \text{as} \quad \|x\| \rightarrow +\infty$) then $F_t(x)$ is norm-coercive as well. Moreover, from \eqref{eq:nonlin_equation} $F_t(x)$ has the same Jacobian matrix as $h(x)$, which is invertible by assumption. Consequently, all conditions of Hadamard-Palais global inverse function theorem \cite[Collorary]{wu1972global} are met, and thus $F_t(x)$ is a diffeomorphism (i.e., it is an invertible map whose inverse is continuously differentiable) at each $t$. Therefore, $F_t(x) = 0$ or equivalently ${F}(t,x) = 0$ has a (single) unique solution $x^\ast(t)$ for each $t$, thus $\alpha(t,x)$ has a unique critical point $x^\ast(t)$ at each $t$. Note that, due to the continuity of ${F}(t,x)$, $x^\ast(t)$ depends continuously on time.  

Next, we will show that the unique critical point of $\alpha(t,x)$ at $t$, i.e., $x^\ast(t)$, is indeed the global maximum point of $\alpha(t,x)$ at $t$. In this regard, we consider the second derivative test on the critical point's trajectory, i.e., $x^\ast(t)$, of $\alpha(t,x)$. From \eqref{eq:grad_alph} and followed by matrix differentiation rules \cite{van2010consistent} we can obtain the Hessian matrix of $\alpha(t,x)$, i.e., $\Hes(t,x) \coloneqq \frac{\partial}{\partial x} \left( \nabla_x \alpha(t,x) \right)$ as ($\otimes$ is the Kronecker product):
\begin{align}\label{eq:second_deriv_alph}
	&\Hes(t,x) =   \dfrac{\partial}{\partial x} \left(  J^\top(x) \right) \left( \left[ \gamma(t,x) \, e^{\nu \alpha(t,x)} \right] \otimes I_n  \right)   \\
	& + J^\top(x) \;  \dfrac{\partial \, \gamma(t,x) }{\partial x} \,  e^{\nu \alpha(t,x)} + J^\top(x) \;   \gamma(t,x) \; \dfrac{\partial \, e^{\nu \alpha(t,x)} }{\partial x}. \nonumber 
\end{align}
Recall that on the critical point's trajectory we have $\gamma(t,x^\ast(t))  = 0$, hence, evaluating \eqref{eq:second_deriv_alph} on $x^\ast(t)$ gives: 
\begin{equation*} \label{eq:hess_at_critical_point}
	\Hes(t,x^\ast(t)) = J^\top(x^\ast(t)) \;  \left. \dfrac{\partial}{\partial x} \left(\gamma(t,x)\right) \right|_{x=x^\ast(t)}   e^{\nu \alpha(t,x^\ast(t))}.
\end{equation*}
From \eqref{eq:gamma_i_funnel}, one can get: 
\begin{equation*}\label{gamma_x_at_critical_point}
	\left. \dfrac{\partial}{\partial x} \left(\gamma(t,x)\right) \right|_{x=x^\ast(t)} = \Gamma(t,x^\ast(t)) \;  J (x^\ast(t)),
\end{equation*}
where $\Gamma(t,x^\ast(t)) \in \R^{n \times n}$ is a negative definite diagonal matrix whose diagonal entries are given by: 
\begin{equation*}
	-\nu \left[ e^{-\nu (h_i(x^\ast(t)) - \rl_i(t))} +  e^{-\nu ( \ru_i(t) - h_i(x^\ast(t)) )} \right].
\end{equation*} 
Therefore:
\begin{equation*} \label{eq:hessian_final}
	\Hes(t,x^\ast(t)) = J^\top(x^\ast(t))  \;  \Gamma(t,x^\ast(t)) \; J(x^\ast(t))   \;   e^{\nu \alpha(t,x^\ast(t))}.
\end{equation*}
Notice that $e^{\nu \alpha(t,x^\ast(t))}>0$ for all times. Since $J(x^\ast(t))$ is a full rank square matrix and $\Gamma(t,x^\ast(t))$ is a negative definite matrix, one can infer that the Hessian matrix $\Hes(t,x^\ast(t))$ is negative definite \cite{horn2012matrix}. Therefore, the critical point $x^\ast(t)$ is a local maximizer. Since $x^\ast(t)$ is the unique critical point of $\alpha(t,x)$ at time $t$, we conclude that $x^\ast(t)$ is indeed the (unique) global maximizer of $\alpha(t,x)$.

\section{Proof of Theorem \ref{th:main}}
\label{appen:proof_theorem}

The proof is comprised of two phases. First, we show that there exists a unique and maximal solution $x: [0, \taum) \rightarrow \R^n$ such that $\alphah(t,x(t;x_0))$ remains within $\Omega_{\alphah} \coloneqq (-1, 1)$. Next, by contradiction we show that $\taum=\infty$ holds and eventually conclude that for all times $\alphah(t,x(t;x_0))$ remains strictly in a compact subset of $(-1, 1)$ (i.e., forward completeness) and all the closed-loop signals remain bounded. In the sequel, for brevity in the notations, we drop the dependence on time and/or states whenever it does not cause any ambiguity. 

Differentiating $\alphah(t,x)$ in \eqref{eq:normal_alpha} yields:
\begin{equation} \label{eq:dot_alphh}
	\dot{\alphah} = \frac{2}{\rhod} \left[ \frac{\partial \alpha}{\partial x} \dot{x} + \frac{\partial \alpha}{\partial t} -\frac{1}{2} \left( \drhos + \drhod \, \alphah \right) \right].
\end{equation}
Owing to \eqref{eq:mapped_alphah}, \eqref{eq:xi}, and \eqref{eq:dot_alphh}, we also see:
\begin{equation}\label{eq:epi_dot}
	\depialph = \dfrac{\partial \T}{\partial \alphah} \dot{\alphah} = \xialph \left[ \frac{\partial \alpha}{\partial x} \dot{x} + \frac{\partial \alpha}{\partial t} -\frac{1}{2} \left( \drhos + \drhod \, \alphah \right) \right].
\end{equation}

\textbf{\textit{Phase I.}} First, define the stacked vector $z \coloneqq [x^\top, \alphah]^\top \in \R^{n+1}$ and let $\dot{z} = B(t,z)$ with $B(t,z) \coloneqq [B_x^{\top}(t,x,\alphah), B_{\alphah}(t,x,\alphah)]^\top \in \R^{n+1}$, where $\dot{x} = B_x(t,x,\alphah)$ and $\dot{\alphah} = B_{\alphah}(t,x,\alphah)$. From \eqref{eq:sys_dynamics_firstorder}, \eqref{eq:PPC_control_explicit} and \eqref{eq:mapped_alphah} one can get: 
\begin{equation*}\label{eq:B_x}
	B_x(t,x,\alphah) \coloneqq f(x) - k \, \xialph(t,\alphah) \, \epialph(\alphah) g(x) \frac{\partial \alpha(t,x)}{\partial x}^\top \!\!\!\! + w(t),
\end{equation*}
which along with \eqref{eq:dot_alphh} yields: 
\begin{align*}\label{eq:B_alphah}
	B_{\alphah}(t,x,\alphah) \coloneqq \frac{2}{\rhod(t)} &\left[\frac{\partial \alpha(t,x)}{\partial x} B_x(t,x,\alphah) + \frac{\partial \alpha(t,x)}{\partial t} \right. \\
	&\left. -\frac{1}{2} \left( \drhos(t) + \drhod(t) \, \alphah \right) \right].
\end{align*}

By assumption $x_0$ is such that $\rlalph(0) < \alpha(0,x_0) \leq \alphaopt(0) < \rualph(0)$, thus from \eqref{eq:normal_alpha} $\alphah_0 \coloneqq \alphah(0,x_0) \in \Omega_{\alphah} = (-1,1)$, where  $\Omega_{\alphah}$ is nonempty and open. Now define $\Ox(t) \coloneqq \{ x \in \R^n \mid -1 < \alphah(t,x) < 1\}$, which due to \eqref{eq:normal_alpha} can be re-written as:
\begin{equation}\label{eq:Omega_x}
	\Ox(t) = \{ x \in \R^n \mid \rlalph(t) < \alpha(t,x) < \rualph(t)\}.
\end{equation}
Note that $\Ox(t)$ is nonempty for each $t\geq 0$ since $\alpha(t,x)$ is continuous and $\rlalph(t) <  \alphaopt(t) <\rualph(t)$ holds for all times by construction. Define $\alpha_{t}(x) \coloneqq \alpha(t,x)$. Since a coercive function is also norm-coercive, From Assumption \ref{assum:coercive_alphabar} and \cite[Proposition 1.2]{de1994global} we know $\alpha_{t}(x)$ is a proper  continuous map in $x$ for all $t \geq 0$ \cite[Definition 8.5]{terrell2009stability}. Define $\Omega_{\alpha}(t) \coloneqq \{\rlalph(t) \leq \alpha(t,x) \leq \rualph(t)\}$, which is a compact set for all $t \geq 0$. Since $\alpha_{t}(x)$ is proper the preimage $\alpha_t^{-1}(\Omega_{\alpha}(t))$ = $\cl(\Ox(t)) = \{ x \in \R^n \mid \rlalph(t) \leq \alpha(t,x) \leq \rualph(t)\}$ is compact, thus $\Ox(t)$ is bounded and open for all $t\geq 0$. Now define $\Oxs \coloneqq \bigcup_{t=0}^{+\infty} \Ox(t) \subset \R^n$, which is the (time-invariant) super set containing $\Ox(t), \forall t \geq 0$. Owing to the properties of $\Ox(t)$ established above, $\Oxs$ is nonempty, open, and bounded. Consequently, we can define the open, bounded, and nonempty set $\Oz \coloneqq \Oxs \times \Oalphh$, which does not depend on $t$. It holds that $z_0 \coloneqq [x_0^\top, \alphah_0]^\top \in \Oz$. 

Finally, one can verify that $B(t,z) \equiv B(t,x,\alphah)$ is locally Lipschitz on $z = [x^\top,\alphah]^\top$ over the set $\Oz$ and is (piece-wise) continuous on $t$. Therefore, the hypotheses of Theorem 54 in \cite[p.~476]{sontag1998mathematical} hold and the existence and uniqueness of a maximal solution $z(t;z_0) \in \Oz$ for a time interval $t \in [0, \taum)$ is guaranteed. Accordingly, for all $t \in [0, \taum)$ we get $x(t;x_0) \in \Oxs$ and $\alphah(t,x(t;x_0)) \in \Oalphh$, which indicate boundedness of $x(t;x_0)$ and $\alphah(t,x(t;x_0))$ for all $t \in [0, \taum)$. 

\textbf{\textit{Phase II)}} Notice that owing to $\alphah(t,x(t;x_0)) \in \Oalphh, \forall t \in [0, \taum)$, \eqref{eq:normal_alpha} implies that $\rlalph(t) < \alpha(t,x(t;x_0)) < \rualph(t), \forall t \in [0, \taum)$. On the other hand, since $\alphaopt(t)$ represents the maximum  value of $\alpha(t,x)$ for all times it is ensured that $\alpha(t,x(t;x_0)) \leq \alphaopt(t), \forall t \in [0, \taum)$. Therefore: 
\begin{equation} \label{eq:alpha_bounds_taumax}
\rlalph(t) < \alpha(t,x(t;x_0)) \leq  \alphaopt(t)  < \rualph(t), 
\end{equation} 
for all $t \in [0, \taum)$. In what follows we establish that $\taum = + \infty$ by contradiction. Let $\taum$ be finite, hence, $\alphah(t,x(t;x_0))$ should approach $+1$ or $-1$ as $t \rightarrow \taum$, which, from \eqref{eq:normal_alpha}, yields $\alpha(t,x(t;x_0)) \rightarrow \rualph(t)$ or $\rlalph(t)$. Nevertheless, from \eqref{eq:alpha_bounds_taumax} the only possible case is $\alpha(t,x(t;x_0)) \rightarrow \rlalph(t)$ as $t \rightarrow \taum$. Recall that owing to $\alphah(t,x(t;x_0)) \in \Oalphh, \forall t \in [0, \taum)$, $\epialph$ in \eqref{eq:mapped_alphah} is well-defined for all $t \in [0, \taum)$. Hence, the barrier function $V(\epialph) = 0.5 \epialph^2$, introduced in Section \ref{subsec:control_design}, can be considered as a positive definite and radially unbounded Lyapunov function candidate w.r.t. $\epialph$. Note that, $\alpha(t,x(t;x_0)) \rightarrow \rlalph(t)$ as $t \rightarrow \taum$ implies $V(\epialph) \rightarrow +\infty$ as $t \rightarrow \taum$. Differentiating $V$ w.r.t. time and substituting \eqref{eq:sys_dynamics_firstorder}, \eqref{eq:PPC_control_explicit}, and \eqref{eq:epi_dot} yields:
\begin{align} \label{eq:lyap1}
	\dot{V} &= \epialph \, \xialph \, \rondalph \, g(x) \, u + \epialph \, \xialph \left[\rondalph \left(f(x) + w(t) \right)  + \frac{\partial \alpha}{\partial t} \right. \nonumber \\ 
	& \quad \left. -\frac{1}{2} \left( \drhos + \drhod \, \alphah \right) \right] = -k \, \rondalph \, g(x) \, \rondalph^\top \epialph^2 \, \xialph^2 + \epialph \, \xialph \, \Phi \nonumber \\
	&\leq -k \, \lambda \, \|\gradxalph\|^2 \, \epialph^2 \, \xialph^2 + |\epialph| \, |\xialph| \, |\Phi|,
\end{align}
where $\Phi \coloneqq \rondalph (f(x) + w(t))  + \tfrac{\partial \alpha}{\partial t} -\tfrac{1}{2} ( \drhos + \drhod \, \alphah )$ and $\lambda \coloneqq \lambda_{min}(g_s(x)), \forall x \in \Oxs$, that is the minimum eigenvalue of $g_s(x)$ for all $x \in \Oxs$. Note that, $\lambda > 0$ due to Assumption \ref{assum:symm_g}. From \textit{Phase I}, we know that $\alphah(t,x(t))$ is bounded for all $t \in [0, \taum)$. Moreover, $\drhos(t)$ and $\drhod(t)$ are bounded by construction and $\|w(t)\|$ is bounded by assumption. Due to continuity of $f(x)$ and $x(t) \in \cl(\Oxs)$ for all $t \in [0, \taum)$, where $\cl(\Oxs)$ is a compact set, $\|f(x)\|$ is bounded for all $t \in [0, \taum)$. From \eqref{eq:grad_alph} and \eqref{smooth_alph}, one can verify that $\tfrac{\partial \alpha(t,x)}{\partial x}$ is continuous on $x$ and $t$ and since $\ru_i(t), \rl_i(t)$ in \eqref{eq:predicate_const_rep} are bounded, $\tfrac{\partial \alpha(t,x)}{\partial x}$ is bounded in $t$. In addition, the boundedness of $\tfrac{\partial \alpha(t,x)}{\partial x}$ for all $t \in [0, \taum)$ is guaranteed since from \textit{ Phase I} $x(t)$ is confined in the compact set $\cl(\Oxs)$ for all $t \in [0, \taum)$. Similarly, using \eqref{smooth_alph} (or \eqref{smooth_alph_alt}) and invoking boundedness of $\drl_i(t), \dru_i(t)$ one can also show that $\tfrac{\partial \alpha(t,x)}{\partial t}$ is a bounded function for all $t \in [0, \taum)$. Overall, the above arguments ensure that $|\Phi|$ is bounded for all $t \in [0, \taum)$. 

Note that under Assumption \ref{assu:alpha_globalmax}, $\|\gradxalph(t,x)\| = 0$ if and only if $\alpha(t,x) = \alphaopt(t)$. Moreover, recall that, from Subsection \ref{subsec:single_funnel_const}, $\alphaopt(t) - \rlalph(t) \geq \underline{\varsigma} > 0$ holds $\forall t \geq 0$ by construction. Consequently, by continuity of $\gradxalph(t,x)$, there exists $\epsilon_{\alpha} > 0$ such that $\|\gradxalph(t,x(t;x_0))\| \geq \epsilon_{\alpha} > 0$ when $\alpha(t,x(t;x_0)) \rightarrow \rlalph(t)$ as $t \rightarrow \taum$. Additionally, owing to $\alphah(t,x(t;x_0)) \in \Oalphh, \forall t \in [0, \taum)$ and the boundedness of $\rhod(t) \geq \delta > 0, \forall t\geq 0$ (which holds by construction, see Subsection \ref{subsec:single_funnel_const}), from \eqref{eq:xi} one can infer that $\xialph \geq  \frac{4}{\max(\rhod(t))} > 0$, $\forall t \in [0, \taum)$. Now let $0 < \theta < k \lambda$ and $\sigma \coloneqq  k \lambda - \theta > 0$ be positive constants. Adding and subtracting $\theta  \|\gradxalph\|^2 |\xialph|^2 |\epialph|^2$ in the right-hand side of \eqref{eq:lyap1} yields:
\begin{align} \label{eq:lyap2}
	\dot{V} &\leq -\sigma \, \|\gradxalph\|^2 \, \epialph^2 \, \xialph^2 - |\epialph| \, |\xialph| \left( \theta \|\gradxalph\|^2  |\epialph|  |\xialph| - |\Phi| \right) \nonumber \\
	& = -\sigma \, \|\gradxalph\|^2 \, \epialph^2 \, \xialph^2, \quad \forall |\epialph| \geq \frac{|\Phi|}{\theta {\|\gradxalph\|^2}  |\xialph|}. 
\end{align}
which, indicates that $\epialph$ is Uniformly Ultimately Bounded (UBB) \cite[Theorem 4.18]{khalil2002noninear} when $\alpha(t,x(t;x_0)) \rightarrow \rlalph(t)$ as $t \rightarrow \taum$. Consequently, there exists an ultimate bound $\epialphb > 0$ independent of $\taum$ such that $|\epialph| \leq \epialphb$ as $t \rightarrow \taum$. Hence, $V(\epialph)$ remains bounded and $V(\epialph) \nrightarrow +\infty$ as $t \rightarrow \taum$, which is a contradiction, therefore, $\taum = +\infty$.

Notice that, since $V(\epialph)$ does not tend to infinity when $\alpha(t,x(t;x_0)) \rightarrow \rlalph(t)$ then $\alpha(t,x(t;x_0))$ is kept strictly away from $\rlalph(t)$ for all times. Consequently, since $\rualph(t) - \alphaopt(t) \geq \bar{\varsigma} > 0$ holds by construction, from \eqref{eq:alpha_bounds_taumax}, one can infer that $\alpha(t,x(t;x_0))$ is kept strictly away from $\rlalph(t)$ and $\rualph(t), \forall t \geq 0$, and thus the satisfaction of the $\alpha$-funnel constraint \eqref{eq:alpha_funnelconst} is ensured $\forall t \geq 0$. Recall that, from \eqref{eq:Omega_x} the satisfaction of \eqref{eq:alpha_funnelconst} implies $x(t;x_0) \in \Ox(t)$, which ensures $x(t;x_0)$ remains bounded as $\Ox(t)$ is a bounded set $\forall t \geq 0$. Finally, since $\alpha(t,x(t;x_0))$ is kept strictly away from $\rlalph(t)$ and $\rualph(t), \forall t \geq 0$, there exists constants $-1 < \underline{b} < \bar{b} < 1$ such that $\alphah(t,x(t;x_0)) \in [\underline{b} , \bar{b}] \subset \Oalphh = (-1,1)$, $\forall t \geq 0$. Therefore, from \eqref{eq:mapped_alphah} and \eqref{eq:xi}, $\epialph$ and $\xialph$ are bounded, and thus $u$ in \eqref{eq:PPC_control_explicit} is also bounded $\forall t \geq 0$. 

\end{document}